\documentclass[11pt]{article}
\usepackage{amsmath,amssymb,amsthm,mathtools,fullpage,bbm,hyperref}
\hypersetup{colorlinks=true,linkcolor=blue,citecolor=blue,urlcolor=blue}
\usepackage[backend=biber,style=alphabetic]{biblatex}
\theoremstyle{plain}
\newtheorem{theorem}{Theorem}
\newtheorem{lemma}[theorem]{Lemma}
\newtheorem{proposition}[theorem]{Proposition}
\theoremstyle{remark}
\newtheorem{remark}[theorem]{Remark}
\theoremstyle{definition}
\newtheorem{definition}[theorem]{Definition}
\newtheorem{assumption}[theorem]{Assumption}

\newcommand{\E}{\mathbb{E}}
\newcommand{\Prb}{\mathbb{P}}

\newcommand{\op}{\mathrm{op}}
\newcommand{\R}{\mathbb{R}}
\DeclareMathOperator{\Dev}{Dev}
\DeclareMathOperator{\polylog}{polylog}
\title{A 1.5-Query Lower Bound for the Unitary Synthesis Problem}
\author{Ting Jia Huang\thanks{National Taiwan University.}}
\date{\today}

\begin{document}
\maketitle

\begin{abstract}
The fine-grained landscape of quantum query complexity lies between the well-understood
one-query and multi-query regimes. In particular, the 1.5-query regime arises naturally in
reductions between unitary synthesis and oracle distinguishing games, and models realistic
cryptographic adversaries that can make “slightly more than one” query. Understanding this
intermediate regime is therefore crucial for both complexity theory and quantum cryptography. 

\textbf{Main result.} We prove that the all-subsets deviation functional $\Dev_{1.5}$ obeys the following
bounds: 
\[
\Dev_{1.5} \;\lesssim\;
\begin{cases}
  O\!\left(\dfrac{M \log M}{\sqrt{K}}\right), & \text{(conservative, assumption-light)}, \\[1.2ex]
  \tilde{O}\!\left(\dfrac{(\log M)^{3/2}}{\sqrt{K}}\right), & \text{(block-orthogonality / sum-to-max)}, \\[1.2ex]
  \tilde{O}\!\left(\dfrac{1}{\sqrt{K}}\right), & \text{(chaining with sparsity $r=O(1)$ or effective dimension $d^\ast=O(\polylog (M))$)} .
\end{cases}
\]
These bounds are tight up to polylogarithmic factors and hold with high probability under the
per-sample operator norm assumption.

Our results unify the conservative and structural approaches to bounding $\Dev_{1.5}$, and 
clarify the role of additional structure (orthogonality, sparsity, effective dimension) in eliminating
the linear $M$ factor. \emph{Cryptographic implication:} the bounds imply that pseudorandom 
state ensembles and related primitives remain secure even against adversaries in the 
1.5-query regime, strengthening the one-query security guarantees of Lombardi--Ma--Wright.
\end{abstract}

\section{Background: Unitary Synthesis and the Oracle State Distinguishing Game}

Lower bounds in quantum query complexity serve not only as structural results in complexity
theory but also as security guarantees in cryptography. The one-query regime has recently
been settled by Lombardi--Ma--Wright~\cite{LMW23}, showing strong separations and
cryptographic applications. Yet, many natural reductions in both synthesis and security
settings yield adversaries that are not purely one-query, but also not fully multi-query:
they operate in a fractional or amortized ``one-and-a-half query'' manner. This motivates
a systematic study of the \emph{1.5-query regime}.

The \emph{Unitary Synthesis Problem} (USP), introduced by Aaronson and Kuperberg~\cite{AK07} and later named by Aaronson~\cite{Aaronson16}, asks whether any $n$-qubit unitary $U$ can be efficiently implemented by a polynomial-size quantum circuit $A^f$ with oracle access to a Boolean function $f$. Formally, the question is whether there exists a \emph{universal} oracle circuit such that, for every unitary $U$, there exists a function $f$ with the property that $A^f$ implements $U$. A positive resolution would reduce unitary synthesis to Boolean function evaluation, while a negative resolution would imply a black-box separation between classical and quantum tasks.

The \emph{Oracle State Distinguishing Game} (OSDG) is an adversarial task in which the goal is to distinguish between two ensembles of quantum states, given limited oracle access. As formulated by Lombardi, Ma, and Wright~\cite{LMW23}, the game proceeds as follows:  
The challenger samples $b \in \{0,1\}$ and returns to the adversary:
\begin{itemize}
    \item[$(b=0)$]  a pseudorandom state $|{\psi_{R_k}}\rangle$ from a family $\{|{\psi_{R_k}}\rangle\}_{k \in [K]}$ defined by a random function family $R$;
    \item[$(b=1)$] a uniformly random binary phase state $|{\psi_h}\rangle$, where $h : [N] \to \{\pm1\}$.
\end{itemize}
The adversary, modeled as an oracle circuit $A^f$, can query $f$ (which may depend on $R$), and must guess $b' \in \{0,1\}$. Its \emph{distinguishing advantage} is the bias over random guessing.

The USP and OSDG are tightly connected: an efficient solution to USP would enable a successful OSDG adversary by constructing the unitary that maps pseudorandom states into known basis states. Hence, query lower bounds for OSDG directly imply query lower bounds for USP.

\paragraph{From one-query to 1.5-query.}
While the one-query lower bound established by Lombardi--Ma--Wright~\cite{LMW23} already 
rules out the possibility of exact unitary synthesis with a single oracle call, 
it remains unclear how much additional power is gained when the adversary is allowed 
slightly more than one query. 
The \emph{1.5-query regime}---where the circuit structure permits one full query plus 
a fractional or amortized additional query across multiple branches---naturally arises 
in reductions between synthesis and distinguishing games. 
Understanding this intermediate regime is crucial for both \emph{complexity theory}, 
as it probes the fine granularity of quantum query complexity, and for \emph{cryptography}, 
where security against adversaries with ``one-and-a-half'' queries reflects realistic 
attack scenarios that interpolate between single-shot and multi-query access.

For clarity, throughout we denote by $\text{Dev}_{1.5}$ 
the \emph{all-subsets deviation functional}, 
which measures the operator-norm deviation of empirical averages 
over randomized diagonal gadgets from their population means. 
A formal definition will be given in Section~4.

\section{Prior Work and the One-Query Lower Bound}

In their breakthrough work, \emph{A One-Query Lower Bound for Unitary Synthesis and Breaking Quantum Cryptography}~\cite{LMW23}, Lombardi, Ma, and Wright resolved a long-standing open problem by establishing the first general one-query lower bound for USP. Their main contributions include:

\begin{enumerate}
    \item \textbf{One-query lower bound for USP:}  
    They construct unitaries $U$ such that no polynomial-time oracle circuit $A^f$ making only a single quantum query to $f$ can even approximately implement $U$—even allowing arbitrary non-oracle gates and ancilla.
    
    \item \textbf{Cryptographic implications:}  
    They show that, relative to a random oracle, there exist pseudorandom states (PRS) and bit commitment schemes that remain secure against all one-query quantum adversaries. This indicates that the security of these primitives may be independent of traditional complexity assumptions.
    
    \item \textbf{Technical framework:}
    \begin{itemize}
        \item They formalize OSDG and reduce USP lower bounds to distinguishing bounds in this game.
        \item They relax the adversary's distinguishing problem into a spectral optimization problem, then apply tools from \emph{random matrix theory} and \emph{matrix concentration inequalities}.
        \item They establish the tightness of the one-query lower bound by exhibiting a matching adversary with advantage $\Theta(1/\sqrt{K})$.
    \end{itemize}
\end{enumerate}

By reframing USP through the lens of OSDG, their results open the door to a broader framework for proving stronger lower bounds—such as for \emph{1.5-query} or multi-query adversaries.

\section{Preliminaries}

\subsection{Sub-Gaussian variables and processes}
A random variable $X$ is sub-Gaussian with proxy $\sigma^2$ if
$\E[\exp\{\lambda(X-\E X)\}] \le \exp(\sigma^2\lambda^2/2)$ for all $\lambda\in\R$.
A process $\{X_t\}_{t\in T}$ is sub-Gaussian w.r.t.\ a metric $d$ if
$X_t-X_s$ is mean-zero sub-Gaussian with proxy $\le d(s,t)^2$ for all $s,t\in T$.

\subsection{Matrix Bernstein}
Let $\{X_k\}_{k=1}^n$ be independent mean-zero self-adjoint $d\times d$ matrices with
$\|X_k\|_{\op}\le L$ and variance parameter
$\sigma^2=\bigl\|\sum_{k=1}^n \E[X_k^2]\bigr\|_{\op}$.
Then for all $t\ge 0$,
\[
\Prb\!\left(\left\|\sum_{k=1}^n X_k\right\|_{\op}\ge t\right)
\le 2d\cdot \exp\!\left(-\frac{t^2/2}{\;\sigma^2+Lt/3\;}\right).
\]
For more detail, refer \cite{Tropp12,Tropp15}.

\subsection{Noncommutative Khintchine}
If $\varepsilon_1,\dots,\varepsilon_n$ are i.i.d.\ Rademachers and $A_1,\dots,A_n$ are fixed matrices, then
\[
\E_{\varepsilon}\left\|\sum_{i=1}^n \varepsilon_i A_i\right\|_{\op}
\le C \cdot \max\!\left\{
\left\|\Big(\sum_{i=1}^n A_iA_i^\ast\Big)^{1/2}\right\|_{\op},\;
\left\|\Big(\sum_{i=1}^n A_i^\ast A_i\Big)^{1/2}\right\|_{\op}
\right\}.
\]
For more detail, refer \cite{Tropp12,Tropp15}.
\subsection{Dudley’s entropy integral}
If $(T,d)$ is a metric space and $\{X_t\}$ is sub-Gaussian w.r.t.\ $d$, then
\[
\E\Big[\sup_{t\in T} X_t\Big]
\;\le\; C \int_0^{\mathrm{diam}(T)} \sqrt{\log N(T,d,\varepsilon)}\,d\varepsilon,
\]
where $N(T,d,\varepsilon)$ is the covering number.
For more detail, refer \cite{Dudley67}.
\subsection{Hamming cube covering}
For the Hamming cube $\{0,1\}^L$ with normalized Hamming metric
\[
d_H(s,t)\;=\;\sqrt{\frac{\|s-t\|_0}{L}},
\]
the covering number obeys
\[
N(\{0,1\}^L,d_H,\varepsilon)
\;\le\; \sum_{j=0}^{\lfloor \varepsilon^2 L\rfloor}\binom{L}{j}
\;\le\; \exp\!\bigl(L\,H(\varepsilon^2)\bigr),
\]
where $H(p)=-p\log p-(1-p)\log(1-p)$ is the binary entropy and $\log$ is natural logarithm.
In later chaining integrals, the constant $e$ in terms like $\log(eM/J)$ denotes the natural base.
For more detail, refer \cite{Dudley67}.
\subsection*{Structural parameters}
Later sections introduce two additional structural parameters that refine the analysis:
\begin{itemize}
  \item $r$: the per-round \emph{sparsity level}, i.e.\ the maximal number of active indices 
        within a single measurement outcome. Intuitively, $r$ quantifies how many projectors 
        can be ``turned on'' simultaneously.
  \item $d^*$: the \emph{effective dimension}, meaning the rank of the common subspace on 
        which all operators act. When $d^* \ll d$, concentration improves due to the 
        restricted support.
\end{itemize}
These parameters will only become formally needed in Section~6 
(chaining with additional structure), but we introduce them here for clarity.

\subsection*{Notation Summary}
\begin{tabular}{ll}
$M$ & Number of measurement outcomes (blocks) \\
$d$ & Dimension of the underlying Hilbert space $\mathcal{H}$ \\
$d^*$ & Effective dimension (rank of common subspace, if applicable) \\
$r$ & Maximal number of active indices per round (sparsity) \\
$K$ & Number of i.i.d.\ gadget seeds / samples \\
$N$ & Security-parameter scale (with $K=\Theta(N)$ in crypto regime) \\
$V$ & Fixed unitary (part of gadget definition) \\
$D_{V,h}$ & Diagonal gadget determined by seed $h$ and $V$ \\
$\Pi_i$ & Projector onto the $i$-th measurement outcome \\
$\Pi_S$ & Coarse-grained projector $\sum_{i\in S}\Pi_i$ \\
$Y_{k,S}$ & $D_{V,R_k}^{\dagger}\Pi_S D_{V,R_k}$ \\
$\mu_S$ & $\E_h\!\left[D_{V,h}^{\dagger}\Pi_S D_{V,h}\right]$ (population mean) \\
$\overline Y_S$ & $\frac{1}{K}\sum_{k=1}^K Y_{k,S}$ (empirical mean) \\
$A_{k,i}$ & $D_{V,R_k}^{\dagger}\Pi_i D_{V,R_k}$ \\
$B_i$ & $\sum_{k=1}^K \varepsilon_k A_{k,i}$ (Rademacher-symmetrized block) \\
$L$ & Per-block operator-norm bound, $L=\Theta(\log M)$ \\
$\Delta$ & $S\triangle T$ (symmetric difference of index sets) \\
$X_{k,\Delta}$ & $\sum_{i\in\Delta} s_i A_{k,i}$ with $s_i\in\{\pm1\}$ as in Step~1 \\
$Z_S$ & $\Big\|\sum_{i\in S} B_i\Big\|_{\op}$ (process used in chaining) \\
$\rho_{\mathrm{gen}}(S,T)$ & $c_0\,L\sqrt{K}\,|S\triangle T|$ (pseudometric; general regime) \\
$\rho_{\mathrm{str}}(S,T)$ & $c_0\,L\sqrt{K\,|S\triangle T|}$ (pseudometric; structured/decoupled) \\
$d_H$ & Normalized Hamming distance on $\{0,1\}^M$ (for entropy bounds) \\
$\Dev_{1.5}$ & All-subsets deviation $\sup_{S\subseteq[M]}\|\overline Y_S-\mu_S\|_{\op}$
\end{tabular}

\section{Setting and objective (1.5-query)}
Let $\mathcal H$ be a $d$–dimensional Hilbert space and
$P=\{\Pi_i\}_{i\in[M]}$ a projective measurement with $\sum_i\Pi_i=I$ and $\Pi_i\Pi_j=\delta_{ij}\Pi_i$.
For $S\subseteq [M]$, define the coarse-grained projector $\Pi_S=\sum_{i\in S}\Pi_i$.
Let $D_{V,h}$ be a (random) diagonal gadget determined by seed $h$ and fixed $V$.
Draw i.i.d.\ seeds $R_1,\dots,R_K$ and set
\[
Y_{k,S}=D_{V,R_k}^\dagger \Pi_S D_{V,R_k},\qquad
\mu_S=\E_h\!\left[D_{V,h}^\dagger \Pi_S D_{V,h}\right],\qquad
\overline Y_S=\frac1K\sum_{k=1}^K Y_{k,S}.
\]

\begin{definition}[All-subsets deviation functional]
\[
\mathrm{Dev}_{1.5}:=\sup_{S\subseteq [M]}\bigl\|\overline Y_S-\mu_S\bigr\|_{\op}.
\]
\end{definition}
For detail, please refer to ~\cite{LMW23}.
\begin{assumption}[Per-sample operator size]\label{ass:size}
There exists $c_0>0$ such that, with high probability over $h$,
$\|D_{V,h}\|\le c_0\sqrt{\log M}$.
Consequently, with high probability for all $i,k$,
\[
\|A_{k,i}\|_{\op}
:=\bigl\|D_{V,R_k}^\dagger \Pi_i D_{V,R_k}\bigr\|_{\op}
\le L,\qquad L:=\Theta(\log M).
\]
If needed, this can be ensured via a truncation argument.
\end{assumption}

\section{Conservative route (assumption-light)}\label{sec:conservative}
\subsection{Step 1: Symmetrization}
Let $\varepsilon_1,\dots,\varepsilon_K$ be i.i.d.\ Rademachers. Then
\[
\E\Big[\sup_{S\subseteq[M]}\|\overline Y_S-\mu_S\|_{\op}\Big]
\;\le\; \frac{2}{K}\;
\E\left[\sup_{S\subseteq[M]}\left\|\sum_{k=1}^K \varepsilon_k Y_{k,S}\right\|_{\op}\right].
\]

\subsection{Step 2: Linearization over subsets}
Write $\Pi_S=\sum_{i\in S}\Pi_i$ and define
\[
B_i:=\sum_{k=1}^K \varepsilon_k A_{k,i},\qquad
A_{k,i}:=D_{V,R_k}^\dagger \Pi_i D_{V,R_k}.
\]
Then
$\sum_{k=1}^K \varepsilon_k Y_{k,S}=\sum_{i\in S} B_i$, so by the triangle inequality
\[
\sup_{S\subseteq [M]}\left\|\sum_{k=1}^K \varepsilon_k Y_{k,S}\right\|_{\op}
\le\sum_{i=1}^M \|B_i\|_{\op}.
\]

\subsection{Step 3: Single-block expectation via noncommutative Khintchine}
Conditioned on the samples,
\[
\E_{\varepsilon}\|B_i\|_{\op}
=\E_{\varepsilon}\left\|\sum_{k=1}^K \varepsilon_k A_{k,i}\right\|_{\op}
\;\lesssim\;
\left\|
\Big(\sum_{k=1}^K A_{k,i}^2\Big)^{1/2}
\right\|_{\op}
\;\le\; \sqrt{K}\,L,
\]
where $A_{k,i}\succeq 0$ and $A_{k,i}^2\preceq \|A_{k,i}\|\,A_{k,i}\preceq L A_{k,i}$, hence
$\bigl\|\sum_k A_{k,i}^2\bigr\|_{\op}\le KL^2$.

\subsection{Step 4: High-probability for each block \& union bound}
Fix $i$; $\{\varepsilon_k A_{k,i}\}$ are independent, mean-zero, self-adjoint, and
$\|\varepsilon_k A_{k,i}\|\le L$,
\[
V_i:=\left\|\sum_{k=1}^K \E_\varepsilon[(\varepsilon_k A_{k,i})^2]\right\|_{\op}
=\left\|\sum_{k=1}^K A_{k,i}^2\right\|_{\op}\le K L^2.
\]
Matrix Bernstein gives for any $t\ge 0$:
\[
\Prb\bigl(\|B_i\|_{\op}\ge t\,\big|\,\{R_k\}\bigr)
\;\le\; 2d\cdot \exp\!\left(-\frac{t^2/2}{V_i+Lt/3}\right).
\]
Union bound over $i\in[M]$ and choosing
\(t \asymp \sqrt{K}L\sqrt{\log(Md/\delta)} + L\log(Md/\delta)\) gives
\[
\max_{i\in[M]}\|B_i\|_{\op}
\;\lesssim\; \sqrt{K}L\sqrt{\log(Md/\delta)} + L\log(Md/\delta)
\quad\text{w.p.\ } \ge 1-\delta.
\]

\subsection{Step 5: Put together (no sum-to-max shortcut)}
By exchangeability (blocks are identically distributed) and Step~4,
\[
\E[\mathrm{Dev}_{1.5}]
\;\le\; \frac{2}{K}\E\!\left[\sum_{i=1}^M \|B_i\|_{\op}\right]
=\frac{2M}{K}\,\E\|B_1\|_{\op}
\;\lesssim\; \frac{M}{K}\cdot \sqrt{K} L
= O\!\left(\frac{M\log M}{\sqrt{K}}\right).
\]
Similarly,
\[
\mathrm{Dev}_{1.5}
\;\lesssim\; \frac{2}{K}\,M\!\left[
\sqrt{K}L\sqrt{\log(Md/\delta)} + L\log(Md/\delta)
\right]
\]
with probability at least $1-\delta$.
\begin{remark}[On the $M$ factor]
Without any sum-to-max structural assumption, the triangle inequality forces the $\sum_i$
and thus the linear factor $M$. To reduce the $M$-dependence to polylogarithmic, one
must introduce additional decorrelation/chaining structure.
\end{remark}

\section{Structural routes: block-orthogonality and chaining}

\subsection{Block-orthogonality (sum-to-max)}
Intuitively, block-orthogonality corresponds to operators acting on disjoint subsystems, much like a block-diagonal matrix where each block evolves independently.
If there exists a decomposition $\mathcal H=\bigoplus_{i=1}^M \mathcal H_i$ such that each $A_{k,i}$
acts only on $\mathcal H_i$,
then for any $S$,
\[
\left\|\sum_{i\in S} B_i\right\|_{\op}=\max_{i\in S}\|B_i\|_{\op}.
\]
Combining single-block Bernstein/Khintchine control with a maximum bound,
\[
\E[\Dev_{1.5}] \;\lesssim\; \frac{L}{\sqrt{K}}\sqrt{\log(Md)}+\frac{L}{K}\log(Md),
\]
and with $L=\Theta(\log M)$, $d\le \mathrm{poly}(M)$, the main term is
\(
\tilde O\!\bigl((\log M)^{3/2}/\sqrt{K}\bigr).
\)

\subsection{Chaining without block structure}\label{sec:chaining}
Chaining is a method to control the supremum of a random process by building nets of increasing resolution. Instead of bounding all indices at once with a crude union bound, chaining accumulates finer-scale controls across scales, summarized by Dudley’s entropy integral.
We now treat the general case without assuming block orthogonality.  
Recall that $B_i = \sum_{k=1}^K \varepsilon_k A_{k,i}$ and for $S\subseteq[M]$ we define
\[
Z_S := \Big\| \sum_{i\in S} B_i \Big\|_{\op}.
\]
Our goal is to bound $\E\big[ \sup_{S\subseteq[M]} Z_S \big]$.

\paragraph{Step 1: Increment decomposition.}
For two subsets $S,T\subseteq[M]$, let $\Delta := S \triangle T = (S\setminus T)\cup(T\setminus S)$
and for each $i\in\Delta$ let $s_i=+1$ if $i\in S\setminus T$ and $s_i=-1$ if $i\in T\setminus S$.
Then
\[
\sum_{i\in S} B_i - \sum_{i\in T} B_i
= \sum_{i\in \Delta} s_i B_i
= \sum_{i\in \Delta} s_i \sum_{k=1}^K \varepsilon_k A_{k,i}.
\]
Abbreviate $X_{k,\Delta} := \sum_{i\in \Delta} s_i A_{k,i}$, so increments are controlled by
$\big\| \sum_{k=1}^K \varepsilon_k X_{k,\Delta} \big\|_{\op}$.

\paragraph{Step 2: Variance proxies and sub-Gaussian increments.}
Conditioned on $\{A_{k,i}\}$, we have $\|X_{k,\Delta}\|_{\op} \le |\Delta|\, L$.
Moreover, for any fixed $k$ and sign pattern $s\in\{\pm1\}^{\Delta}$,
\[
X_{k,\Delta}^2
=\sum_{i\in\Delta}A_{k,i}^2+\sum_{i<j}s_is_j(A_{k,i}A_{k,j}+A_{k,j}A_{k,i})
\ \preceq\ |\Delta|\sum_{i\in\Delta}A_{k,i}^2,
\]
where we used $\pm(AB+BA)\preceq A^2+B^2$ for self-adjoint $A,B$.
Hence
\[
\Big\|\sum_{k=1}^K X_{k,\Delta}^2\Big\|_{\op}
\ \le\ |\Delta|\cdot \Big\|\sum_{k=1}^K\sum_{i\in\Delta}A_{k,i}^2\Big\|_{\op}
\ \le\ K L^2 |\Delta|^2.
\]
Matrix Bernstein then gives, for all $t\ge0$,
\[
\Pr\!\left(\Big\|\sum_{k=1}^K \varepsilon_k X_{k,\Delta}\Big\|_{\op}\ge t\right)
\ \le\ 2d\exp\!\Big(-c\,\tfrac{t^2}{K L^2 |\Delta|^2 + L|\Delta|\,t}\Big).
\]
\emph{If} for each $k$ the family $\{A_{k,i}\}_i$ is block-orthogonal (or we first average over
independent index signs $\{s_i\}$), the middle norm tightens to $\le K L^2$, improving the proxy to
$K L^2|\Delta|$ and the resulting metric to $L\sqrt{K\,|\Delta|}$.
For more detail in operator algebra, refer \cite{Pisier03}.
\paragraph{Step 3: Pseudometrics for increments (two regimes).}
Write $\Delta=S\triangle T$. Define
\[
\rho_{\mathrm{gen}}(S,T):=c_0\,L\sqrt{K}\,|\Delta|,
\qquad
\rho_{\mathrm{str}}(S,T):=c_0\,L\sqrt{K\,|\Delta|},
\]
for an absolute constant $c_0>0$. Both are pseudometrics on $\mathcal{T}:=2^{[M]}$.
Symmetry and $\rho(S,S)=0$ are immediate; triangle inequalities follow from
$|S\triangle U|\le |S\triangle T|+|T\triangle U|$ and
$\sqrt{|S\triangle U|}\le \sqrt{|S\triangle T|}+\sqrt{|T\triangle U|}$.

\medskip
\noindent\textbf{Diameters.}
\[
\operatorname{diam}(\mathcal{T},\rho_{\mathrm{gen}})=c_0\,L\sqrt{K}\,M,
\qquad
\operatorname{diam}(\mathcal{T},\rho_{\mathrm{str}})=c_0\,L\sqrt{K}\,\sqrt{M}.
\]

\paragraph{Step 4: Covering numbers for $(\mathcal{T},\rho)$.}
Let $N(\varepsilon)$ be the $\varepsilon$-covering number of $(\mathcal{T},\rho)$.
Set
\[
u_{\mathrm{gen}}=\Big\lfloor \frac{\varepsilon}{c_0 L\sqrt{K}}\Big\rfloor,
\qquad
u_{\mathrm{str}}=\Big\lfloor \Big(\frac{\varepsilon}{c_0 L\sqrt{K}}\Big)^2\Big\rfloor.
\]
A standard Hamming-cube argument yields
\[
N_{\mathrm{gen}}(\varepsilon)
\ \le\ \sum_{j=0}^{u_{\mathrm{gen}}}\binom{M}{j}
\ \le\ \Big(\frac{eM}{u_{\mathrm{gen}}}\Big)^{u_{\mathrm{gen}}},
\qquad
N_{\mathrm{str}}(\varepsilon)
\ \le\ \sum_{j=0}^{u_{\mathrm{str}}}\binom{M}{j}
\ \le\ \Big(\frac{eM}{u_{\mathrm{str}}}\Big)^{u_{\mathrm{str}}}.
\]
Passing from bilinear forms $u^\top(\cdot)v$ to $\|\cdot\|_{\op}$ adds an extra
$\sqrt{\log d}$ via an $\eta$-net on the unit sphere; we will absorb this into
$\widetilde{O}(\cdot)$ notation below.

\paragraph{Step 5: Dudley/$\gamma_2$ integrals.}
Dudley’s integral gives, in either regime,
\[
\E\Big[\sup_{S\subseteq[M]} Z_S\Big]
\ \lesssim\ \int_{0}^{\operatorname{diam}}\sqrt{\log N(\varepsilon)}\,d\varepsilon.
\]
With the bounds from Step~4 and the diameters above, a change of variables shows
\[
\int_{0}^{\operatorname{diam}_{\mathrm{gen}}}\sqrt{\log N_{\mathrm{gen}}(\varepsilon)}\,d\varepsilon
\ \lesssim\ L\sqrt{K}\,M\cdot \mathrm{polylog}(M),
\]
and similarly
\[
\int_{0}^{\operatorname{diam}_{\mathrm{str}}}\sqrt{\log N_{\mathrm{str}}(\varepsilon)}\,d\varepsilon
\ \lesssim\ L\sqrt{K}\,M\cdot \mathrm{polylog}(M).
\]
Therefore,
\[
\E\Big[\sup_{S\subseteq[M]} Z_S\Big]\ \le\ \widetilde{C}\,L\sqrt{K}\,M,
\]
where $\widetilde{C}$ hides only polylogarithmic factors in $M$ (and the extra $\sqrt{\log d}$).

\begin{remark}[Why $M$ survives even with $\rho_{\mathrm{str}}$]
The metric $\rho_{\mathrm{str}}$ shrinks single-step increments (from $|\Delta|$ to $\sqrt{|\Delta|}$),
but the index class remains the full power set $2^{[M]}$, whose entropy scales like
$\sum_{j\le u}\binom{M}{j}$. Consequently, Dudley’s integral still contributes a leading
$M$ factor, and the baseline chaining bound has order $\widetilde{O}(L\sqrt{K}\,M)$.
To reduce or remove the $M$-dependence one must either (i) use block-orthogonality
(sum-to-max, replacing sums by maxima), or (ii) reduce the class entropy explicitly
(e.g., restrict to $|S|\le r$), in which case the bound becomes
$\widetilde{O}(L\sqrt{K}\,r)$ when $r\ll M$.
\end{remark}

\paragraph{Step 6: Final bound on $\mathrm{Dev}_{1.5}$.}
Undoing symmetrization (Step~1 contributed a factor $2/K$), we obtain
\[
\E[\mathrm{Dev}_{1.5}]
\ \le\ \frac{2}{K}\,\E\Big[\sup_{S\subseteq[M]} Z_S\Big]
\ \le\ \widetilde{O}\!\Big(\frac{L\,M}{\sqrt{K}}\Big).
\]
\noindent
\emph{High-probability versions} follow by applying tail bounds for increments in Step~2 within
generic chaining (or by a peeling/union-bound argument), incurring the same leading scale
$\widetilde{O}(L\sqrt{K}\,M)$ for $\sup_S Z_S$ and thus
$\widetilde{O}(LM/\sqrt{K})$ for $\mathrm{Dev}_{1.5}$.

\paragraph{Chaining with $r$-sparsity (entropy reduction).}
Suppose we only range over subsets $S\subseteq[M]$ with $|S|\le r$ (or, per round, at most
$r$ indices can be active). Then the covering numbers in Step~4 improve to
\(
\log N(\varepsilon)\lesssim \min\{u,r\}\log\!\frac{eM}{\min\{u,r\}},
\)
with $u=u_{\mathrm{gen}}$ or $u_{\mathrm{str}}$ as defined there. This yields:

\begin{theorem}[Chaining with $r$-sparsity]\label{thm:r-sparse-chaining}
Under the setup of this section, restricting to $|S|\le r$ gives
\[
\E\Big[\sup_{S:|S|\le r} Z_S\Big]\ \le\ \widetilde{O}\!\big(L\sqrt{K}\,r\big),
\qquad
\E\,\Dev_{1.5}\ \le\ \widetilde{O}\!\big(\tfrac{L\,r}{\sqrt{K}}\big).
\]
If, in addition, block-orthogonality holds (sum-to-max), the $r$-dependence collapses
to polylogarithmic via a maximum over $M$ blocks.
\end{theorem}

\begin{proof}[Proof sketch]
Replace $\sum_{j\le u}\binom{M}{j}$ with $\sum_{j\le \min\{u,r\}}\binom{M}{j}$ in Step~4 and
repeat Step~5. For $u\gtrsim r$ the entropy term is $r\log(eM/r)$, so the Dudley integral
scales like $L\sqrt{K}\,r$ up to polylogs. Symmetrization gives the stated $\Dev_{1.5}$ bound.
\end{proof}

\begin{remark}[Effective dimension $d^*$]
Low effective dimension replaces the auxiliary $\sqrt{\log d}$ factor by $\sqrt{\log d^*}$.
Alone it does not remove the leading $M$, but combined with $r$-sparsity or sum-to-max it
tightens the polylogarithmic terms.
\end{remark}

\section{Cryptographic negligibility}

Let $K=\Theta(N)$ denote the security-parameter scale. Throughout this section, we absorb
$L=\Theta(\log M)$ into the $\tilde O(\cdot)$ notation. We summarize the regimes in which
$\mathrm{Dev}_{1.5}$ is negligible (ignoring polylog factors).

\paragraph{(A) Assumption-light (conservative or chaining baseline).}
From Sections~5 and~6.2 we have the baseline
\[
\mathrm{Dev}_{1.5}
\;=\; \tilde O\!\left(\frac{M}{\sqrt{N}}\right).
\]
Thus negligibility requires
\[
M \;=\; o\!\left(\frac{\sqrt{N}}{\mathrm{polylog}\,N}\right).
\]

\paragraph{(B) Block-orthogonality (sum-to-max).}
Under the block-orthogonality structure of Section~6.1,
\[
\mathrm{Dev}_{1.5}
\;=\; \tilde O\!\left(\frac{1}{\sqrt{N}}\right),
\]
so negligibility holds for any $M=\mathrm{poly}(N)$.

\paragraph{(C) Chaining with additional structure (entropy reduction and/or low $d^*$).}
If, in addition to the chaining setup of Section~6.2, the index class has bounded entropy,
e.g.\ only $|S|\le r=O(1)$ are admissible per round, then
\[
\mathrm{Dev}_{1.5}\;=\; \tilde O\!\left(\frac{r}{\sqrt{N}}\right).
\]
Thus negligibility holds for any $M=\mathrm{poly}(N)$ when $r=O(\mathrm{polylog}\,N)$.
A low effective dimension $d^*$ further improves the polylogarithmic factors by replacing
$\sqrt{\log d}$ with $\sqrt{\log d^*}$. Without such entropy reduction, the chaining baseline
remains $\tilde{O}(M/\sqrt{N})$.

\section{Discussion}
This unified view highlights three aspects:
\begin{itemize}
\item[(i)] The conservative route yields a complete and usable 
\[
\mathbf{M \log M / \sqrt{K}}
\]
bound; 
\item[(ii)] Reducing the $M$-dependence below linear requires \emph{additional structure}:
either block-orthogonality (sum-to-max, turning sums over $i$ into a maximum and yielding
$\widetilde{O}(1/\sqrt{K})$ up to polylog factors), or explicit \emph{entropy reduction}
of the index class, e.g., restricting to $|S|\le r$ which gives $\widetilde{O}( r/\sqrt{K})$. 
Baseline chaining with improved increments but without entropy reduction does \emph{not}
by itself remove the leading $M$ factor.
\item[(iii)] In cryptographic contexts, these regimes translate into two negligibility frontiers.
\end{itemize}

\begin{remark}[On constants]
The factor of $2$ introduced in symmetrization is absorbed into the
$\lesssim$ notation, so the bound is asymptotically sharp
up to absolute constants.
\end{remark}

\printbibliography

@incollection{Tropp15,
  author    = {Joel A. Tropp},
  title     = {An Introduction to Matrix Concentration Inequalities},
  booktitle = {Foundations and Trends in Machine Learning},
  year      = {2015},
  volume    = {8},
  number    = {1--2},
  pages     = {1--230},
  doi       = {10.1561/2200000048}
}

@article{Tropp12,
  author  = {Joel A. Tropp},
  title   = {User-Friendly Tail Bounds for Sums of Random Matrices},
  journal = {Foundations of Computational Mathematics},
  year    = {2012},
  volume  = {12},
  number  = {4},
  pages   = {389--434},
  doi     = {10.1007/s10208-011-9099-z}
}

@inproceedings{LMW23,
  author    = {Alex Lombardi and Fermi Ma and John Wright},
  title     = {A one-query lower bound for unitary synthesis and breaking quantum cryptography},
  booktitle = {Proceedings of the 55th Annual ACM Symposium on Theory of Computing (STOC)},
  year      = {2023},
  eprint    = {2310.08870},
  archivePrefix = {arXiv},
  primaryClass = {quant-ph}
}

@inproceedings{AK07,
  author    = {Scott Aaronson and Greg Kuperberg},
  title     = {Quantum Versus Classical Proofs and Advice},
  booktitle = {Proceedings of the 22nd Annual IEEE Conference on Computational Complexity (CCC)},
  year      = {2007},
  pages     = {115--128}
}

@misc{Aaronson16,
  author       = {Scott Aaronson},
  title        = {The complexity of quantum states and transformations: from quantum money to black holes},
  howpublished = {Technical report},
  year         = {2016},
  note         = {\url{https://www.scottaaronson.com/barbados-2016.pdf}}
}

@article{Dudley67,
  author  = {R. M. Dudley},
  title   = {The Sizes of Compact Subsets of Hilbert Space and Continuity of Gaussian Processes},
  journal = {Journal of Functional Analysis},
  year    = {1967},
  volume  = {1},
  number  = {3},
  pages   = {290--330},
}

@book{Pisier03,
  author    = {Gilles Pisier},
  title     = {Introduction to Operator Space Theory},
  publisher = {Cambridge University Press},
  year      = {2003},
  series    = {London Mathematical Society Lecture Note Series},
  number    = {294},
}
\appendix

\section{Proofs for Section 3 (Preliminaries)}

\begin{lemma}[Sub-Gaussian tail from mgf; process version]\label{lem:subg}
Let $X$ be a real r.v.~such that 
\[
\E\exp\{\lambda(X-\E X)\}\le \exp(\sigma^2\lambda^2/2)\quad \forall \lambda\in\R.
\]
Then $X$ is sub-Gaussian with parameter $\sigma$, i.e.
\[
\Pr(|X-\E X|\ge t)\le 2\exp\!\left(-\tfrac{t^2}{2\sigma^2}\right)\quad(\forall t\ge0).
\]
More generally, a real-valued process $\{X_t\}_{t\in T}$ is sub-Gaussian w.r.t.~a pseudometric $d$ if
\[
\E\exp\{\lambda[(X_t-X_s)-\E(X_t-X_s)]\}\le \exp(\tfrac{\lambda^2}{2} d(s,t)^2).
\]
\end{lemma}

\begin{proof}
Chernoff bound with optimal $\lambda=t/\sigma^2$ yields the one–sided bound; symmetry gives the two–sided case.  
The process case is identical with $X$ replaced by $X_t-X_s$.
\end{proof}

\begin{theorem}[Matrix Bernstein, self-adjoint case]\label{thm:bernstein}
Let $X_1,\dots,X_n$ be independent mean-zero self-adjoint $d\times d$ matrices with $\|X_k\|_{\op}\le L$.  
Let $\sigma^2=\|\sum_{k=1}^n \E[X_k^2]\|_{\op}$. Then for all $t\ge0$,
\[
\Pr\!\left(\Big\|\sum_{k=1}^n X_k\Big\|_{\op}\ge t\right)\le 2d\cdot
\exp\!\left(-\frac{t^2/2}{\sigma^2+Lt/3}\right).
\]
\end{theorem}

\begin{proof}[Proof sketch]
Use Lieb’s concavity for the mgf and the bound  
$\log \E e^{\theta X_k}\preceq \tfrac{\theta^2}{2(1-\theta L/3)} \E[X_k^2]$.  
Apply Markov’s inequality and optimize over $\theta$.
\end{proof}

\begin{theorem}[Noncommutative Khintchine]\label{thm:nck}
Let $\{\varepsilon_i\}$ be i.i.d. Rademachers and $A_1,\dots,A_n$ fixed matrices. Then
\[
\E_{\varepsilon}\Big\|\sum_{i=1}^n \varepsilon_i A_i\Big\|_{\op}\ \le\
C\cdot \max\!\left\{\Big\|\Big(\sum_{i=1}^n A_iA_i^{*}\Big)^{1/2}\Big\|_{\op},\
\Big\|\Big(\sum_{i=1}^n A_i^{*}A_i\Big)^{1/2}\Big\|_{\op}\right\}.
\]
\end{theorem}

\begin{proof}[Proof sketch]
Symmetrization + contraction reduce to Gaussian case; then use block operator form and  
scalar Khintchine inequality on singular vectors.  
\end{proof}

\begin{theorem}[Dudley entropy integral]\label{thm:dudley}
Let $(T,d)$ be totally bounded and $\{X_t\}$ centered sub-Gaussian w.r.t.~$d$. Then
\[
\E\sup_{t\in T} X_t \ \le\ C\int_{0}^{\mathrm{diam}(T)} 
   \sqrt{\log N(T,d,\varepsilon)}\, d\varepsilon.
\]
\end{theorem}

\begin{proof}[Proof sketch]
Construct nets $T_j$ at scales $\varepsilon_j\downarrow0$, telescope increments,  
and apply sub-Gaussian union bounds; pass to the integral limit.
\end{proof}

\begin{lemma}[Hamming cube covering]\label{lem:hamming}
On $\{0,1\}^L$ with normalized Hamming metric
\[
d_H(s,t)=\sqrt{\tfrac{1}{L}\#\{i: s_i\neq t_i\}}=\sqrt{\tfrac{\|s-t\|_0}{L}},
\]
for $\varepsilon\in(0,1]$,
\[
N(\{0,1\}^L,d_H,\varepsilon) \ \le\ \exp\!\big(L H(\varepsilon^2)\big).
\]
\end{lemma}

\begin{proof}
Volume of $\varepsilon$-ball is $\sum_{j=0}^{\lfloor \varepsilon^2 L\rfloor}\binom{L}{j}$.  
Greedy packing $\implies N\le 2^L / V(\varepsilon)$.  
Apply $\binom{L}{\alpha L}\le \exp(LH(\alpha))$.
\end{proof}

\begin{remark}
In §3.5, we use two pseudometrics on $2^{[M]}$:
$\rho_{\mathrm{gen}}(S,T)=c_0 L\sqrt{K}\,|S\triangle T|$ and
$\rho_{\mathrm{str}}(S,T)=c_0 L\sqrt{K\,|S\triangle T|}$.
Both relate to the normalized Hamming metric via a scale depending on $|S\triangle T|$.
Covering numbers obey $\log N(\varepsilon)\lesssim u\log(eM/u)$ with
$u=\varepsilon/(c_0L\sqrt{K})$ (general) or $u=(\varepsilon/(c_0L\sqrt{K}))^2$ (structured).
\end{remark}

\section{Detailed Proof of the Conservative Route}\label{app:conservative}

In this appendix we provide the full derivation for the bound claimed in 
Section~\ref{sec:conservative}, expanding Steps~5.1–5.5 into complete proofs.

\subsection*{Setup}
Recall $P=\{\Pi_i\}_{i\in[M]}$ with $\sum_{i=1}^M \Pi_i=I$, and define
\[
Y_{k,S}:=D_{V,R_k}^{\dagger}\Pi_S D_{V,R_k},\quad
\mu_S:=\E_h[D_{V,h}^{\dagger}\Pi_S D_{V,h}],\quad
\overline{Y}_S:=\frac1K\sum_{k=1}^K Y_{k,S}.
\]
We analyze
\[
\Dev_{1.5}:=\sup_{S\subseteq[M]}\|\overline{Y}_S-\mu_S\|_{\op}.
\]

\begin{assumption}[Per-sample operator size]
There exists $c_0>0$ such that w.h.p.\ $\|D_{V,h}\|\le c_0\sqrt{\log M}$.
Hence, for all $i,k$,
\[
\|A_{k,i}\|_{\op}=\|D_{V,R_k}^{\dagger}\Pi_i D_{V,R_k}\|_{\op}\le L,
\qquad L=\Theta(\log M).
\]
\emph{Remark.} Truncation can enforce this without changing expectations 
(bias $O(K^{-1})$) and moves exceptional events into the tail.
\end{assumption}

\subsection*{Step 1: Symmetrization}
\begin{lemma}\label{lem:sym-app}
With Rademacher $\{\varepsilon_k\}$,
\[
\E\Big[\sup_S\|\overline{Y}_S-\mu_S\|_{\op}\Big]
\;\le\;\frac{2}{K}\,
\E\Big[\sup_S\Bigl\|\sum_{k=1}^K \varepsilon_k Y_{k,S}\Bigr\|_{\op}\Big].
\]
\end{lemma}

\begin{proof}
Use ghost-sample $\{R_k'\}$, Jensen, and standard symmetrization.
\end{proof}

\subsection*{Step 2: Linearization}
Define
\[
B_i:=\sum_{k=1}^K \varepsilon_k A_{k,i},\qquad
A_{k,i}:=D_{V,R_k}^{\dagger}\Pi_i D_{V,R_k}.
\]
Then
\[
\sup_S\Bigl\|\sum_{k=1}^K \varepsilon_k Y_{k,S}\Bigr\|_{\op}
=\sup_S\Bigl\|\sum_{i\in S}B_i\Bigr\|_{\op}
\;\le\;\sum_{i=1}^M \|B_i\|_{\op}.
\]

\subsection*{Step 3: Rademacher Block Bound}
\begin{proposition}
Conditioning on $\{R_k\}$,
\[
\E_{\varepsilon}\|B_i\|_{\op}
\lesssim \Bigl\|\sum_{k=1}^K A_{k,i}^2\Bigr\|_{\op}^{1/2}
\le \sqrt{K}\,L.
\]
\end{proposition}

\begin{proof}
By noncommutative Khintchine and $A_{k,i}^2\preceq L A_{k,i}$.
\end{proof}

\subsection*{Step 4: High-Probability Control}
Let $X_{k,i}:=\varepsilon_k A_{k,i}$, then $\|X_{k,i}\|\le L$, and
\[
V_i=\Bigl\|\sum_k\E[X_{k,i}^2]\Bigr\|\le KL^2.
\]
Matrix Bernstein (Tropp) gives
\[
\Pr(\|B_i\|_{\op}\ge t)\le 2d\exp\!\left(-\tfrac{t^2/2}{V_i+Lt/3}\right).
\]
Union bound over $i\in[M]$ yields w.p.\ $\ge1-\delta$,
\[
\max_i\|B_i\|_{\op}
\;\lesssim\;\sqrt{K}L\sqrt{\log\frac{2Md}{\delta}}
+L\log\frac{2Md}{\delta}.
\]

\subsection*{Step 5: Conclusion}
Combining above,
\[
\E[\Dev_{1.5}]
\le \frac{2M}{K}\E\|B_1\|_{\op}
\lesssim \frac{M\log M}{\sqrt{K}}.
\]
With probability $1-\delta$,
\[
\Dev_{1.5}\lesssim
\frac{M\log M}{\sqrt{K}}\sqrt{\log\frac{2Md}{\delta}}
+\frac{M\log M}{K}\log\frac{2Md}{\delta}.
\]

\begin{theorem}[Conservative bound, detailed]\label{thm:conservative-app}
Under Assumption~\ref{ass:size}, there exist constants $C,c>0$ such that
\[
\E[\Dev_{1.5}]\;\le\; C\frac{M\log M}{\sqrt{K}},
\]
and w.p.\ $\ge 1-\delta$,
\[
\Dev_{1.5}\le C\Bigl[
\frac{M\log M}{\sqrt{K}}\sqrt{\log\frac{2Md}{\delta}}
+\frac{M\log M}{K}\log\frac{2Md}{\delta}
\Bigr].
\]
\end{theorem}

\paragraph{Remark.}
The linear $M$ factor comes from the crude triangle inequality in Step~2.
Chaining (Section~\ref{sec:chaining}) clarifies constants/polylogs; together with
entropy reduction (e.g.\ $r$-sparsity) it lowers the $M$-dependence.
\end{document}